\documentclass[letterpaper,conference,10pt]{ieeeconf}
\IEEEoverridecommandlockouts
\makeatletter
\let\NAT@parse\undefined
\makeatother

\usepackage[square, sort&compress, numbers]{natbib}

\usepackage{amsmath,graphicx,epsfig,color,amsfonts, algorithm, subfigure}
\usepackage{color,soul}
\usepackage{balance}

\graphicspath{{./fig/}}

\usepackage{version,xspace}
\usepackage[noend]{algorithmic}

\def\real{\mathbb{R}}

\newcommand{\until}[1]{\{1,\dots, #1\}}
\newcommand{\subscr}[2]{#1_{\textup{#2}}}

\newcommand{\setdef}[2]{\{#1 \; | \; #2\}}

\newcommand{\map}[3]{#1: #2 \rightarrow #3}

\newcommand{\argmin}{\operatorname{argmin}}

\newcommand\oprocendsymbol{\hbox{$\square$}}
\newcommand\oprocend{\relax\ifmmode\else\unskip\hfill\fi\oprocendsymbol}

\renewcommand{\baselinestretch}{0.99}

\renewcommand{\theenumi}{\roman{enumi}}

\newcommand\bit[1]{\textit{\textbf{#1}}}

\def \bs {\boldsymbol}
\def \mc {\mathcal}

\newtheorem{theorem}{Theorem}
\newtheorem{proposition}{Proposition}
\newtheorem{lemma}{Lemma}
\newtheorem{corollary}{Corollary}

\newtheorem{remark}{Remark}

\newtheorem{assumption}{Assumption}

\renewcommand{\theenumi}{(\roman{enumi}}

\title{On Epidemic Spreading under Mobility on Networks
\thanks{This work was supported by ARO grant W911NF-18-1-0325.}
}
\author{Vishal Abhishek$^{1}$ and Vaibhav Srivastava$^{2}$
\thanks{$^{1}$Vishal Abhishek is with the Department of Mechanical Engineering,
       Michigan State University,
       East Lansing, MI 48824-1226, USA
        {\tt\small abhishe3 at egr.msu.edu}}%
\thanks{$^{2}$Vaibhav Srivastava is with the Electrical and Computer Engineering, Michigan State University, East Lansing, MI 48824-1226, USA
        {\tt\small vaibhav at egr.msu.edu}}%
}

\begin{document}

\maketitle
\thispagestyle{empty}
\pagestyle{empty}

\begin{abstract}
We study a coupled epidemic-mobility model in which, at each time, individuals move in a network of spatially-distributed regions  (sub-populations) according to a Continuous Time Markov Chain (CTMC) and subsequently interact with the local sub-population according to an SIS model. We derive a deterministic continuum limit model describing these interactions. 
We prove the existence of a disease-free equilibrium and an endemic equilibrium under different parameter regimes and establish their (almost) global asymptotic stability using Lyapunov techniques. For the stability of disease-free equilibrium, we also deduce some simple sufficient conditions which highlight the influence of mobility on the behavior of the SIS dynamics. Finally, we numerically illustrate that the derived model provides a good approximation to the stochastic model with a finite population and also demonstrate the influence of the graph structure on the transient performance.
\end{abstract}

\section{INTRODUCTION}

Contagion processes describe spread of entities such as influence, disease, and rumors through a network. In a simple contagion process, exposure to a single contagious individual may spread the contagion, while in a complex contagion, multiple exposures are required~\cite{min2018competing}. Epidemic propagation models belong to the class of simple contagion processes and have been used to model disease spread~\cite{anderson1992infectious, DE-JK:10}, spread of computer viruses~\cite{kleinberg2007computing, wang2009understanding}, routing in mobile communication networks~\cite{zhang2007performance}, and spread of rumors~\cite{jin2013epidemiological}. These models are also closely related to branching processes~\cite{DE-JK:10} that have been used to model social contagion and cascade behavior~\cite{watts2002simple}. 

When epidemic propagation occurs over a network of spatially distributed regions, the movement of individuals across regions influences the epidemic propagation. In this paper, we study a coupled epidemic-mobility model in which individuals may travel across regions according to a mobility model and subsequently affect the epidemic propagation at those regions. Using Lyapunov techniques, we characterize the steady state behavior of the model under different parameter regimes and characterize the influence of mobility on epidemic dynamics. 



Epidemic models have been extensively studied in the literature. SIS (Susceptible-Infected-Susceptible) and SIR (Susceptible-Infected-Recovered)  models are two most widely studied models, wherein individuals are classified into one of the three categories: susceptible, infected or recovered. Classical SIS/SIR models study the dynamics of the fraction of the population in each category~\cite{DE-JK:10}. Network models consider sub-populations clustered into different nodes, and the sub-population-level dynamics is determined by the local SIS/SIR interactions as well as the interactions with neighboring sub-populations in the network graph~\cite{AnalysisandControlofEpidemics_ControlSysMagazine_Pappas, meiBullo2017ReviewPaper_DeterministicEpidemicNetworks, fall2007epidemiological, khanafer_Basar2016stabilityEpidemicDirectedGraph, meiBullo2017ReviewPaper_DeterministicEpidemicNetworks}. Authors in~\cite{Hassibi2013globaldynEpidemics,Ruhi_Hassibi2015SIRS} study network epidemic dynamics in discrete time setting. 

Several generalizations of the SIR/SIS models have been proposed including, SIER model~\cite{ AnalysisandControlofEpidemics_ControlSysMagazine_Pappas, mesbahi2010graph}, where an additional classification ``exposed" is introduced, SIRS~\cite{DE-JK:10,Ruhi_Hassibi2015SIRS}, where individuals get temporary immunity after recovery and then become susceptible again, and SIRI~\cite{gomez2015abruptTransitionsSIRI, pagliara_NaomiL2018bistability}, where after recovery, agents become susceptible with a different rate of infection. The network epidemic dynamics have also been studied for time-varying networks~ \cite{bokharaie2010_EpidemicTVnetwork, Preciado2016_EpidemicTVnetwork,  Beck2018_EpidemicTimeVaryingNetwork}.


Epidemic spread under mobility has been modeled and analyzed as reaction-diffusion process in \cite{colizza2008epidemicReaction-DiffusionMetapopuln, Saldana2008continoustime_Reaction-DiffnMetapopln}. Epidemic spread with mobility on a multiplex network of sub-populations has been modeled and studied in \cite{soriano2018spreading_MultiplexMobilityNetwork_Metapopln}. Authors of this work consider a discrete model in which, at each time, individuals randomly move to another node, participate in epidemic propagation and then return to their home node. 

In this paper, we study a coupled epidemic-mobility model comprised of a set of sub-populations located in a network of  spatially distributed regions. Individuals within each sub-population (region) can travel across regions according to a Continuous Time Markov Chain (CTMC) and upon reaching a new region participate in the local SIS epidemic process. We extend the results for the deterministic network SIS model~\cite{fall2007epidemiological, Hassibi2013globaldynEpidemics, khanafer_Basar2016stabilityEpidemicDirectedGraph, meiBullo2017ReviewPaper_DeterministicEpidemicNetworks} to the proposed model and characterize its steady state and stability properties.  

The major contributions of this paper are fourfold. First, we derive a deterministic continuum limit model describing the interaction of the SIS dynamics with the Markovian mobility dynamics.  We discuss the connections of the derived model with existing population-level models that capture spatial aspects. Second, we rigorously characterize the existence and stability of the equilibrium points of the derived model under different parameter regimes. Third, for the stability of the disease-free equilibrium, 
we determine some useful sufficient conditions which highlight the influence of mobility on the steady state behavior of the dynamics. Fourth, we numerically illustrate that the derived model is a good approximation to the stochastic model with a finite population. We also illustrate the influence of the network topology on the transient properties of the model. 

    
    
    




The remainder of this paper is organized in the following way. In Section \ref{Sec: Mobility Modeling as Continous-Time Markov Process}, we derive the epidemic model under mobility as a continuum limit to two interacting stochastic processes. In Section \ref{sec: analysis}, we characterize the existence and stability of disease-free and endemic equilibrium for the derived model. In Section \ref{Sec: numerical studies}, we illustrate our results using numerical examples. Finally, we conclude in Section \ref{Sec: conclusions}. 

\medskip

\noindent
{\it Mathematical notation:}
For any two real vectors $\bs x$, $\bs y \in \real^n$, we denote:\\
$\bs x \gg \bs y$, if $x_i > y_i$ for all $i \in \until{n}$,\\
$\bs x \geq \bs y$, if $x_i \geq y_i$ for all $i \in \until{n}$,\\
$\bs x > \bs y$, if $x_i \geq y_i$ for all $i \in \until{n}$ and $\bs x \neq \bs y$.\\
For a square matrix $G$, radial abscissa $\map{\mu}{\real^{n\times n}}{\real}$ is defined by 
\[
\mu(G) = \max \setdef{\mathrm{Re}(\lambda)}{\lambda \text{ is an eigenvalue of $G$}}, 
\]
where $\mathrm{Re}(\cdot)$ denotes the real part of the argument. 
Spectral radius $\rho$ is defined by
\[
\rho(G) = \max \setdef{|\lambda|}{\lambda \text{ is an eigenvalue of $G$}}, 
\]
where $|(\cdot)|$ denotes the absolute value of the argument.
For any vector $\bs x = [x_1,\dots,x_n]^\top$, $X=\operatorname{diag}(\bs x)$ is a diagonal matrix with $X_{ii}=x_i$ for all $i \in \until{n}$.

\section{SIS Model under Markovian Mobility} \label{Sec: Mobility Modeling as Continous-Time Markov Process}

We consider $n$ sub-populations of individuals that are located in distinct spatial regions. We assume the individuals within each sub-population can be classified into two categories: (i) susceptible, and (ii) infected. Let $p_i \in [0,1]$ (respectively, $1-p_i$) be the fraction of infected (respectively, susceptible) individuals within sub-population $i \in \until{n}$. We assume that the individuals within each sub-population can travel to regions associated with other sub-populations. Let the connectivity of these regions be modeled by a digraph $\mc G = (\mc V, \mc E)$, where $\mc V =\until{n}$ is the node set and $\mc E \subset \mc V \times \mc V$ is the edge set. We model the mobility of individuals on graph $\mc G$ using a Continuous Time Markov Chain (CTMC) with a stationary generator matrix $Q$, whose $(i,j)$-th entry is $q_{ij}$. The entry $q_{i j} \ge 0$, $i \ne j$, is the instantaneous transition rate from node $i$ to node $j$, and $-q_{ii}= \nu_{i}$ is the total rate of transition out of node $i$, i.e., $\nu_{i} = \sum_{j \ne i}q_{i j}$. Here, $q_{ij} >0$, if $(i,j) \in \mc E$; and $q_{ij}=0$, otherwise. Let $x_{i}(t) \in (0,1)$ be the fraction of the total population that constitutes the sub-population at node $i$ at time $t$. It follows that $\sum_{i=1}^{n}{x_i} = 1$.  Define $\bs p := [p_{1},\dots,p_n]^\top$ and $\bs x := [x_{1},\dots,x_n]^\top$.

We model the interaction of mobility with the epidemic process as follows. At each time $t$, individuals at each node move on graph $\mc G$ according to the CTMC with generator matrix $Q$ and interact with individuals within their current node according to an SIS epidemic process.  For the epidemic process at node $i$, let $\beta_i >0$ and $\delta_i \ge 0$ be the infection and recovery rate, respectively. We let $B >0$ and $D \ge 0$ be the positive and non-negative diagonal matrices with entries $\beta_i$ and $\delta_i$, $i \in \until{n}$, respectively. Similarly we define $P$ as a diagonal matrix with entries $p_i$. We now derive the continuous time dynamics that captures the interaction of mobility and the SIS epidemic dynamics.




\begin{proposition}[\bit{SIS model under mobility}]\label{prop:model}
The dynamics of the fractions of the infected sub-population $\bs p$ and the fractions of the total population $\bs x$ that constitute the sub-population at each node under Markovian mobility model with generator matrix $Q$, and infection and recovery matrices $B$ and $D$, respectively, are
\begin{subequations} \label{eq_Model}
\begin{align}
    \dot{\bs{p}} & = (B-D-L(\bs{x}))\bs{p} - P B \bs{p} \label{eq_p}\\ 
    \dot{\bs{x}} & = Q^\top \bs{x}, \label{eq_x}
\end{align}
\end{subequations}
where $L(\bs{x})$ is a matrix with entries 
\[
l_{ij}(\bs x) = \begin{cases} \sum_{j\neq i}q_{j i} \frac{x_{j}}{x_{i}}, & \text{if } i = j, \\
-q_{j i} \frac{x_{j}}{x_{i}}, & \text{otherwise}.
\end{cases}
\]
\end{proposition}
\medskip

\begin{proof}
Consider a small time increment $h>0$ at time $t$. Then the fraction of the total population present at node $i$ after the evolution of CTMC in time-interval $[t, t+h)$ is
\begin{equation} \label{eq_popln}
   x_{i}(t+h)= x_{i}(t)(1-\nu_{i}h)+ \displaystyle\sum_{j\neq i}q_{j i} x_{j}(t)h + o(h) .
\end{equation}
Individuals within each node interact according to SIS dynamics. Thus, the fraction of infected population present at node $i$ is: 
\begin{multline} \label{eq_SIS_No of Infected}
\! \! \! \! \! \! x_{i}(t+h) p_{i}(t+h)= - x_{i}(t) \delta_{i} p_{i}(t)h + x_{i}(t)\beta_{i} p_{i}(t)(1-p_{i}(t))h \\
   + x_{i}(t) p_{i}(t)(1-\nu_{i}h) + \displaystyle\sum_{j\neq i}q_{j i} p_{j}(t) x_{j}(t)h + o(h). 
\end{multline}
The first two terms on the right side of \eqref{eq_SIS_No of Infected} correspond to epidemic process within each node, whereas the last two terms correspond to infected individuals coming from other nodes due to mobility. Using the expression of $x_{i}$ from \eqref{eq_popln} in \eqref{eq_SIS_No of Infected} and taking the limit $h \to 0^+$ gives
\begin{multline} \label{eq_SIS_pi}
   \dot{p}_{i}= - \delta_{i} p_{i} + \beta_{i} p_{i}(1-p_{i})
   +\displaystyle\sum_{j\neq i}q_{j i} (p_{j}-p_{i}) \frac{x_{j}}{x_{i}} .
\end{multline}
Similarly taking limits in \eqref{eq_popln} yields
\begin{equation} \label{eq_popln_det}
   \dot{x}_{i} = -\nu_{i} x_{i} + \displaystyle\sum_{j\neq i}q_{j i} x_{j}.
\end{equation}
Rewriting \eqref{eq_SIS_pi} and \eqref{eq_popln_det} in vector form establishes the proposition. 
\end{proof}

\begin{remark}[\bit{Comparison with other models}]
The population level epidemic propagation models in theoretical ecology incorporate spatial aspects by using a partial differential equation that is obtained by adding a spatial diffusion operator to the infected population dynamics~\cite{stone2012sir}. Since, Laplacian matrix is a diffusion operator on a graph, dynamics~\eqref{eq_Model} can be interpreted as a network equivalent of the population models with spatial aspects. The dependence of the Laplacian matrix on $\bs x$ in~\eqref{eq_Model} is more general than the constant diffusion coefficient discussed in~\cite{stone2012sir}. \oprocend 
\end{remark}

\section{Analysis of SIS Model under Markovian Mobility} \label{sec: analysis}
In this section, we analyze the SIS model under mobility~\eqref{eq_Model} under the following standard assumption: 
\begin{assumption} \label{Assumption:StrongConnectivity}
 Digraph $\mc G$ is strongly connected which is equivalent to matrix $Q$ being irreducible \cite{Bullo-book_Networks}. \oprocend
\end{assumption}
Let $\bs v$ be the right eigenvector  of $Q^\top$ associated with eigenvalue at $0$. We assume that $\bs v$ is scaled such that its inner product with the associated left eigenvector $\bs 1_{n}$ is unity, i.e., $\bs 1_{n}^\top \bs v = 1$. We call an equilibrium point $(\bs p^*, \bs x^*)$, an endemic equilibrium point, if at equilibrium the disease does not die out, i.e., $\bs p^* \neq 0$, otherwise, we call it a disease-free equilibrium point. Let $L^*:=L(\bs x^*)=L(\bs v)$.


\begin{theorem}[\bit{Existence and Stability of Equilibria}] \label{thm:stability}
For the SIS model under Markovian mobility~\eqref{eq_Model} with Assumption~\ref{Assumption:StrongConnectivity}, the following statements hold 
\begin{enumerate}
    \item if $\bs p(0) \in [0,1]^n$, then $\bs p(t) \in [0,1]^n$ for all $t>0$. Also, if $\bs p(0) > \bs 0_n$, then $\bs p(t) \gg \bs 0_n$ for all $t>0$;
    \item the model admits a disease-free equilibrium at $(\bs p^*, \bs x^*)= (\bs 0_n, \bs v)$; 
    \item the model admits an endemic equilibrium at $(\bs p^*, \bs x^*) = (\bar{\bs p}, \bs v)$, $\bar{\bs p} \gg 0$,  if and only if $\mu (B-D-L^*) > 0$; 
    \item the disease-free equilibrium is globally asymptotically stable if and only if $\mu (B-D-L^*) \leq 0$ and is unstable otherwise;
    \item the endemic equilibrium is almost globally asymptotically stable if $\mu (B-D-L^*) > 0$ with region of attraction $\bs p(0) \in [0,1]^n$ such that $\bs p(0) \neq \bs 0_n$. 
    \end{enumerate}
\end{theorem}
\medskip

\begin{proof}
The first part of statement (i) follows from the fact that $\dot{\bs{p}}$ is either directed tangent or inside of the region $[0,1]^n$ at its boundary which are surfaces with $p_i =0$ or $1$ . For the second part of (i), we rewrite \eqref{eq_p} as:
\begin{equation*}
    \dot{\bs{p}} = (B(I-P)+A(\bs{x}))\bs{p} - E(t) \bs{p}
\end{equation*}
where $L(\bs x)=C(\bs x)-A(\bs x)$ with $C(\bs x)$ composed of the diagonal terms of $L(\bs x)$, $A(\bs x)$ is the non-negative matrix corresponding to the off-diagonal terms, and $E(t)=C(\bs x(t))+D$ is a diagonal matrix. Now, consider a variable change $\bs y(t) := e^{\int_{0}^{t}E(t) dt}\bs p(t)$. The rest of the proof is same as in \cite[Theorem 4.2 (i)]{meiBullo2017ReviewPaper_DeterministicEpidemicNetworks}.\\
The second statement follows by inspection.\\
The proof of the third statement is presented in Appendix~\ref{Appendix: existence of non-trivial eqb}. 

\noindent\textbf{Stability of disease-free equilibria:}
To prove the fourth statement, we begin by establishing sufficient conditions for instability. The linearization of \eqref{eq_Model} at $(\bs p, \bs x) = (\bs 0, \bs v)$ is
\begin{equation} \label{eq_px linear}
    \begin{bmatrix}
     \dot{\bs p} \\
     \dot{\bs x}
     \end{bmatrix} = \begin{bmatrix}
     B-D-L^* & 0_{n\times n} \\
     0_{n\times n} & Q^\top
     \end{bmatrix}\begin{bmatrix}
        \bs p \\
        \bs x
     \end{bmatrix} .
\end{equation}
Since the system matrix in~\eqref{eq_px linear} is block-diagonal, its eigenvalues are the eigenvalues of the block-diagonal sub-matrices. Further, since radial abscissa $\mu(Q^\top)$ is zero, a sufficient condition for instability of the disease-free equilibrium is that $\mu (B-D-L^*) > 0$.

For the case of $\mu (B-D-L^*) \leq 0$, we now show that the disease-free equilibrium is a globally asymptotically stable equilibrium.
Since $(B-D-L^*)$ is an irreducible Metzler matrix with $\mu (B-D-L^*) \leq 0$, there exists a positive diagonal matrix $R$ such that 
\[
R(B-D-L^*)+(B-D-L^*)^ \top R = -K,
\]
where $K$ is a positive semi-definite matrix \cite[Proposition 1 (iv), Lemma A.1]{khanafer_Basar2016stabilityEpidemicDirectedGraph}.  Define $\Tilde{L} := L(\bs x)-L^*$ and $r := \|R\|$, where $\|\cdot\|$  denotes the the induced two norm of the matrix. 

Since $\bs x(0) \gg 0$, under Assumption~\ref{Assumption:StrongConnectivity}, $x_i(t)$ is lower bounded by some positive constant and hence, $\Tilde{L}$ is bounded and continuously differentiable.
Since $\bs x$ is bounded and exponentially converges to $\bs x^*$, it follows that $\|\Tilde{L}(x)\|$ locally exponentially converges to $\|\Tilde{L}(\bs x^*)\| = 0$ and $\int_{0}^{t} \|\Tilde{L}\| d t$ is bounded for all $t>0$.

Consider the Lyapunov-like function $V(\bs p, t) = \bs p^\top R \bs p - 2 n r \int_{0}^{t} \|\Tilde{L}\| d t$. 
It follows from the above arguments that $V$ is bounded. Therefore,
\begin{align}\label{Vdot_trivial}
     \dot{V} & = \bs p^\top R \dot{\bs p} + \dot{\bs p}^\top R \bs p -2 n r \|\Tilde{L}\| \nonumber \\
            & = \bs p^\top (R(B-D-L^*)+(B-D-L^*)^\top R) \bs p \nonumber \\
            & \quad -2 \bs p^\top R (L(\bs x)-L^*)\bs p - 2\bs p^\top R P B \bs p -2 n r\|\Tilde{L}\| \nonumber \\
            & = -\bs p^\top K \bs p -2 \bs p^\top R\tilde{L}(\bs x)\bs p - 2\bs p^\top R P B \bs p \nonumber  \\
            & \quad - 2 n r\|\Tilde{L}\| \nonumber \\
            & \leq -\bs p^\top K \bs p + 2 n r \|\Tilde{L}\| - 2 n r \|\Tilde{L}\| - 2 \bs p^\top R P B \bs p \nonumber  \\
            & \leq - 2 \bs p^\top R P B \bs p \leq 0 .
\end{align}
Since all the signals and their derivatives are bounded, it follows that $\Ddot{V}(t)$ is bounded and hence $\dot{V}$ is uniformly continuous in $t$. Therefore from Barbalat's lemma and its application to Lyapunov-like functions ~\cite[Lemma 4.3, Chapter 4]{slotine1991applied} it follows that $\dot{V} \rightarrow 0$ as $t \rightarrow \infty$. Consequently, from \eqref{Vdot_trivial}, $\bs p^\top R P B \bs p \rightarrow 0$. Since $R > 0$, $B > 0$  and $ p_i \geq 0$, $\bs p(t) \rightarrow \bs 0$ as $t \rightarrow \infty$. This establishes global attractivity of the disease-free equilibrium point. We now establish its stability. 

We note that $\|\Tilde{L}(\bs x)\|$  is a real analytic function of $\bs x$, for $\bs x \gg \bs 0$. Therefore, there exists a region $\|\bs x - \bs x^*\|<\delta_1$ in which $\|\Tilde{L}(\bs x)\|\leq k_1\|\bs x - \bs x^*\|$ for some $k_1>0$. Also, since $\bs x - \bs x^*$ is globally exponentially stable, $\|\bs x(t) - \bs x^*\| \leq  k_2 e^{-\alpha t} \|\bs x(0) - \bs x^*\|$ for some $k_2$, $\alpha >0$. Thus, if $\|\bs x(0) - \bs x^*\| < \frac{\delta_1}{k_2}$, then $\|\Tilde{L}(\bs x)\|\leq k_1 k_2 e^{-\alpha t}\|\bs x(0) - \bs x^*\|$. This implies $\int_{0}^{t} \|\Tilde{L}\| d t \leq \frac{k}{\alpha}\|\bs x(0) - \bs x^*\|$, where $k:=k_1 k_2$.  Now, since $\dot{V}(\bs p, t)\leq 0$, 
\begin{equation*} \label{eq_trivial stability}
\begin{split}
  V(\bs p(0), 0) &= \bs p(0)^\top R \bs p(0) \\
       &\geq V(\bs p(t), t) \\
       &\geq \bs p(t)^\top R \bs p(t) -2\frac{n r k \|\bs x(0)-\bs x^*\|}{\alpha} \\
       &\geq \subscr{R}{min}\|\bs p(t)\|^2 - 2\frac{n r k \|\bs x(0)-\bs x^*\|}{\alpha} ,
  \end{split}
\end{equation*}
  where $\subscr{R}{min} = \min_{i} (R_{i})$. Equivalently, 
\begin{equation*}
\begin{split}
  \|\bs p(t)\|^2 &\leq \frac{r}{\subscr{R}{min}} \|\bs p(0)\|^2 + 2\frac{n r k\|\bs x(0)-\bs x^*\|}{\alpha \subscr{R}{min}}.
  \end{split}
\end{equation*}
It follows using stability of $\bs x$ dynamics, that for any $\epsilon >0$, there exists $\delta >0$ , such that $\| \bs x(0)-\bs x^*\|^2 + \| \bs p(0) \|^2 \leq \delta ^2 \Rightarrow \| \bs p(t)\|^2 + \| \bs x(t)-\bs x^*\|^2 \leq \epsilon ^2$. This establishes stability. Together, global attractivity and stability prove the fourth statement. 

\noindent\textbf{Stability of endemic equilibria:}
Finally, we prove the fifth statement. To this end, we first establish an intermediate result. 
\begin{lemma} \label{Lemma:p_i tends to 0 implies p tends to 0}
For the dynamics~\eqref{eq_p}, if $p_{i}(t) \rightarrow 0$ as $t \rightarrow \infty$, for some $i \in \until{n}$, then $\bs p(t) \rightarrow \bs 0$ as $t \to \infty$.\\
\end{lemma} 

\begin{proof}
The dynamics of $p_i$ are
\begin{equation} \label{pi dot expanded}
    \dot{p_{i}} = (\beta_{i} - \delta_{i}-l_{i i}(\bs x)) p_{i} - \displaystyle\sum_{j\neq i} l_{i j}(\bs x) p_{j} - \beta_{i} p_{i}^2 .
\end{equation}
It can be easily seen that $\Ddot{p}_{i}$ is bounded and hence $\dot{p}_{i}$ is uniformly continuous in $t$. Now if $p_{i}(t) \rightarrow 0$ as $t \rightarrow \infty$, it follows from Barbalat's lemma \cite[Lemma 4.2]{slotine1991applied} that $\dot{p}_{i} \rightarrow 0$. Therefore, from \eqref{pi dot expanded} and the fact that $- l_{i j}(\bs x) \geq 0$ and $p_{i} \geq 0$, it follows that $p_{j}(t) \rightarrow 0$ for all $j$ such that $- l_{i j} (\bs x) \neq 0$. Using Assumption~\ref{Assumption:StrongConnectivity} and applying the above argument at each node implies
$\bs p(t) \rightarrow \bs 0$. 
\medskip
\end{proof}

Define $\Tilde{\bs p} := \bs p-\bs p^*$, $P^* := \operatorname{diag}(\bs p^*)$ and $\Tilde{P} := \operatorname{diag}(\Tilde{\bs p})$. Then
\begin{equation*}
    \begin{split}
        \dot{\Tilde{\bs p}} & =  (B-D-L(\bs x)- P B) \bs p \\
                        & =  (B-D-L^*- P^* B) \bs p^* + (B-D-L^*- P^* B) \Tilde{\bs p} \\
                        & \quad - \Tilde{L}(\bs x) \bs p - \Tilde{P}B \bs p \\
                        & = (B-D-L^*- P^* B) \Tilde{\bs p} -  \Tilde{L}(\bs x) \bs p - \Tilde{P}B \bs p .
    \end{split}
\end{equation*}
where $(B-D-L^*- P^* B) \bs p^* = \bs 0$, as ($\bs p^*$, $\bs x^*$) is an equilibrium point.

Note that $(B-D-L^*- P^* B)$ is an irreducible Metzler matrix. The Perron-Frobenius theorem for irreducible Metzler matrices \cite{Bullo-book_Networks} implies 
$\mu (B-D-L^*- P^* B) = 0$ and the associated eigenvector $\bs p^* \gg \bs 0_n$. Also, this means there exists a positive-diagonal matrix $R_2$ and a positive semi-definite matrix $K_2$ such that
\[
R_{2}(B-D-L^* -P^*B)+(B-D-L^*-P^*B)^\top R_{2} = -K_2 .
\]

Similar to the proof of the fourth statement, take $V_{2}(\tilde{\bs p}, t) = \Tilde{\bs p}^\top R_{2} \Tilde{\bs p} - 2n r_{2} \int_{0}^{t} \|\Tilde{L}\| d t$, where $r_{2} := \|R_{2}\|$.
Then,
\begin{equation*} \label{Vdot_non-trivial}
\begin{split}
           \dot{V_{2}} & = \tilde{\bs p}^\top R_{2} \dot{\tilde{\bs p}} + \dot{\tilde{\bs p}}^\top R_{2} \tilde{\bs p} -2n r_{2} \|\Tilde{L}\| \\
            & = \tilde{\bs p}^\top (R_{2}(B-D-L^*-P^*B)\\
            & +(B-D-L^*-P^*B)^\top R_{2}) \tilde{\bs p} \\
            & -2 \tilde{\bs p}^\top R_{2}\Tilde{L}(\bs x)\bs p - 2 \tilde{\bs p}^\top R_{2} \tilde{P} B \bs p -2n r_{2}\|\Tilde{L}\|\\
            & = -\tilde{\bs p}^\top K_{2} \tilde{\bs p} -2 \tilde{\bs p}^\top R_{2}\Tilde{L}(\bs x)\bs p - 2 \tilde{\bs p}^\top R_{2} \tilde{P} B \bs p -2n r_{2}\|\Tilde{L}\|\\
            & \leq -\tilde{\bs p}^\top K_{2} \tilde{\bs p}  + 2 n r_{2} \|\Tilde{L}\| - 2 n r_{2} \|\Tilde{L}\| - 2 \tilde{\bs p}^\top R_{2} \tilde{P} B \bs p\\
            & \leq - 2 \tilde{\bs p}^\top R_{2} \tilde{P} B \bs p = -2\displaystyle\sum_{i=1}^{n} (R_2)_i\beta_i \tilde{p}_{i}^2 p_i \leq 0 .
\end{split}
\end{equation*}
 It can be easily shown that $\Ddot{V}_{2}$ is bounded implying $\dot{V}_{2}$ is uniformly continuous. Applying Barbalat's lemma \cite[Lemma 4.2]{slotine1991applied} gives $\dot{V}_{2} \rightarrow 0$ as $t \rightarrow \infty$. Now, since  $R_{2}$ and $B$ are positive diagonal matrices this implies that $\tilde{p}_{i} p_{i} \rightarrow 0$, for each $i$. Using Lemma \ref{Lemma:p_i tends to 0 implies p tends to 0}, and the fact that $\bs p= \bs 0$ is an unstable equilibrium for $\mu (B-D-L^*) > 0$, we have $\tilde{\bs p} \rightarrow \bs 0$ as long as $\bs p(0) \neq \bs 0$.
Stability can be established similarly to the disease-free equilibrium case. This concludes the proof of the theorem. 
\end{proof}

\begin{corollary}[\bit{Stability of disease-free equilibria}] \label{cor:dis-free}
For the SIS epidemic model under Markovian mobility~\eqref{eq_Model} with Assumption~\ref{Assumption:StrongConnectivity} and the disease-free equilibrium $(\bs p^*, \bs x^*)= (\bs 0_n, \bs v)$ the following statements hold
\begin{enumerate}
    \item a necessary condition for stability is $\delta_{i} > \beta_{i} - \nu_{i}$, for each $i \in \until{n}$; 
    \item  a necessary condition for stability is that there exists some $i \in \until{n}$ such that $\delta_i \geq \beta_i$; 
    \item a sufficient condition for stability is $\delta_{i} \geq \beta_{i}$, for each $i \in \until{n}$; 
    \item a sufficient condition for stability is 
    \[
    \frac{\lambda_{2}}{\Big(1+\sqrt{1+\frac{\lambda_{2}}{\sum_{i} w_{i}\big(\delta_{i}-\beta_{i}-m\big)}}\Big)^2 n + 1} + m \geq 0,
    \]
    where 
    $\bs w$ is a positive left eigenvector of $L^*$ such that $\bs w^\top L^* = 0$ with $\max_{i} w_{i} = 1$, $m = \min_{i} (\delta_{i}-\beta_{i})$, $W = \operatorname{diag} (\bs w)$, and  $\lambda_{2}$ is the the second smallest eigenvalue of $\frac{1}{2}(W L^* + L^{*\top} W)$. 
\end{enumerate}
\end{corollary}
\begin{proof}
We begin by proving the first two statements. First, we note that $L^*_{ii} = \nu_i$. This can be verified by evaluating $L^*=L(\bs v)$ and utilising the fact that $Q^\top \bs v = \bs 0$. The necessary and sufficient condition for the stability of disease-free equilibrium is $\mu (B-D-L^*) \leq 0$. Since, $B-D-L^*$ is an irreducible Metzler matrix, a necessary condition for $\mu \leq 0$ is that its diagonal terms are strictly negative, i.e., $\beta_i - \delta_i - \nu_i <0$, for each $i \in \until{n}$. This gives the statement (i).

Perron-Frobenius theorem for irreducible Metzler matrices implies that there exists a real eigenvalue equal to $\mu$ with positive eigenvector, i.e.,
$(B-D-L^*)\bs y = \mu \bs y $, where $\bs y \gg \bs 0_n $. Since, $\mu \leq 0$, written component-wise for $i^*$, where $y_{i^*}=\min(y_i)$ :
\begin{align*}
        & (\beta_{i^*} - \delta_{i^*} -\nu_{i^*})y_{i^*} - \displaystyle\sum_{j\neq i^*} l_{i j}y_j \leq 0 \nonumber  \\
        & \Rightarrow (\beta_{i^*} - \delta_{i^*})y_{i^*} \leq (\nu_{i^*} + \displaystyle\sum_{j\neq i^*} l_{i j})y_{i^*} + \displaystyle\sum_{j\neq i^*} l_{i j}(y_j-y_{i^*})\nonumber \\
        & \Rightarrow (\beta_{i^*} - \delta_{i^*})y_{i^*} \leq \displaystyle\sum_{j\neq i^*} l_{i j}(y_j-y_{i^*}) \nonumber \\
        & \Rightarrow \beta_{i^*} - \delta_{i^*} \leq 0 .
    \end{align*}
This proves the statement (ii).

Since, $L^*$ is a Laplacian matrix, if $\delta_i \geq \beta_i$, for each $i \in \until{n}$, from Gershgorin disks theorem \cite{Bullo-book_Networks}, $\mu \leq 0$, which proves the third statement.

For the last statement, we use an eigenvalue bound for perturbed irreducible Laplacian matrix of a digraph~ \cite[Theorem 6]{wu2005bounds}, stated below:

Let $H = L + \Delta$, where $L$ is an $n\times n$ irreducible Laplacian matrix and $\Delta \neq 0$ is a non-negative diagonal matrix, then  
\begin{equation*}
\begin{split}
  \mathrm{Re}(\lambda(H)) \geq \frac{\lambda_{2}}{\Big(1+\sqrt{1+\frac{\lambda_{2}}{\sum_{i} w_{i}\Delta_i}}\Big)^2 n + 1} > 0,
  \end{split}
\end{equation*}
where, $\bs w$ is a positive left eigenvector of $L$ such that $\bs w^\top L = 0$ with $\max_{i} w_{i} = 1$, $W = \operatorname{diag} (\bs w)$, and  $\lambda_{2}$ is the second smallest eigenvalue of $\frac{1}{2}(W L + L^\top W)$.\\
Now, in our case necessary and sufficient condition for stability of disease-free equilibrium is:
\begin{equation*}
\begin{split}
  \mathrm{Re}(\lambda(L^*+D-B)) & = \mathrm{Re}(\lambda(L^*+\Delta + mI)) \\
  & = \mathrm{Re}(\lambda(L^*+\Delta)) + m \geq 0
  \end{split}
\end{equation*}
where, $m = \min_{i} (\delta_{i}-\beta_{i})$ and $\Delta=D-B-mI$. Applying the eigenvalue bound with $H=L^*+\Delta$ gives the sufficient condition (iv). 
\end{proof}

\begin{remark}
It can be shown that $\bs v$ is the left eigenvector  associated with eigenvalue zero for both $Q$ and $L^*$, i.e., $\bs v^\top Q = \bs v^\top L^* = 0$ and thus can be re-scaled to compute $\bs w = \frac{1}{\operatorname{max_i}(v_i)}\bs v$. \oprocend
\end{remark}
\medskip
\begin{remark}
For  a given graph and the associated mobility transition rates in dynamics~\eqref{eq_Model}, let $m = \operatorname{min}_{i} (\delta_{i}-\beta_{i})$ and $i^*= \argmin_{i} (\delta_{i}-\beta_{i})$. Then, there exist $\delta_i$'s, $i\neq i^*$, that satisfy statement (iv) of Corollary \ref{cor:dis-free} if $m > \subscr{m}{lower}$, where 
\[
\subscr{m}{lower}=-\frac{\lambda_2}{4n+1}.
\] \oprocend
\end{remark}

\begin{remark}(\bit{Influence of mobility on stability of disease-free equilibrium.})
The statement (iv) of Corollary~\ref{cor:dis-free} characterizes the influence of mobility on the stability of disease-free equilibria. In particular, $\lambda_2$ is a measure of ``intensity" of mobility and $m$ is a measure of largest deficit in the recovery rate compared with infection rate among nodes. The sufficient condition in statement (iv) states explicitly how mobility can allow for stability of disease-free equilibrium even under deficit in recovery rate at some nodes. \oprocend
\end{remark}

\section{Numerical Illustrations} \label{Sec: numerical studies}

We start with numerical simulation of epidemic model with mobility in which we treat epidemic spread as well as mobility as stochastic processes. We take $20$ simulations with same initial conditions and parameters and take the average of the results. The fraction of infected populations for different cases are shown in Fig.~\ref{fig:Stochastic}. We take a line graph and the mobility transition rates being equal among out going neighbors of a node. The two cases relate to the stable disease-free equilibrium and stable endemic equilibrium respectively. We have chosen heterogeneous curing or infection rates to elucidate the influence of mobility. If the curing rates, infection rates and the initial fraction of infected population is same for all the nodes, mobility does not play any role. The corresponding simulations of the deterministic model as per Proposition \ref{prop:model} are also shown for comparison. Figure~\ref{fig:Stochastic}~(a) corresponds to the case $\delta_i \geq \beta_i$ for each $i$, whereas Fig.~\ref{fig:Stochastic}~(c) corresponds to the case $\delta_i< \beta_i$ for each $i$. The results support statements (iii) and (ii) of Corollary \ref{cor:dis-free} and lead to, respectively, the stable disease-free equilibrium and the stable endemic equilibrium.

\begingroup
\centering
\begin{figure}[ht!]
\centering
\subfigure[Stable disease-free equilibrium: Stochastic model]{\includegraphics[width=0.23\textwidth]{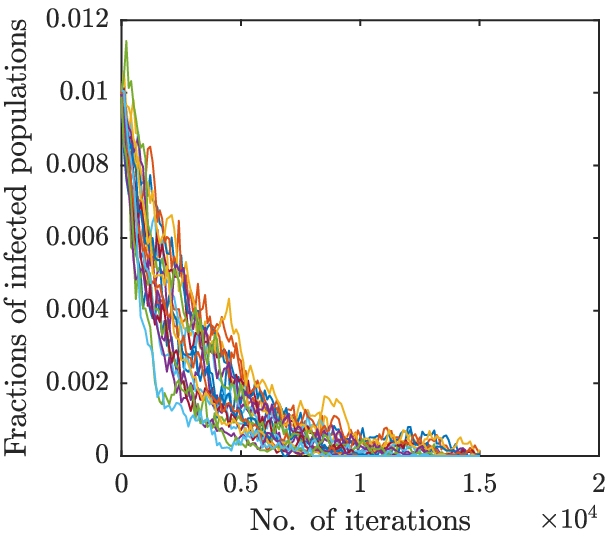}}\label{fig:Stochastic_a}
\subfigure[Stable disease-free equilibrium: Deterministic model]{\includegraphics[width=0.23\textwidth]{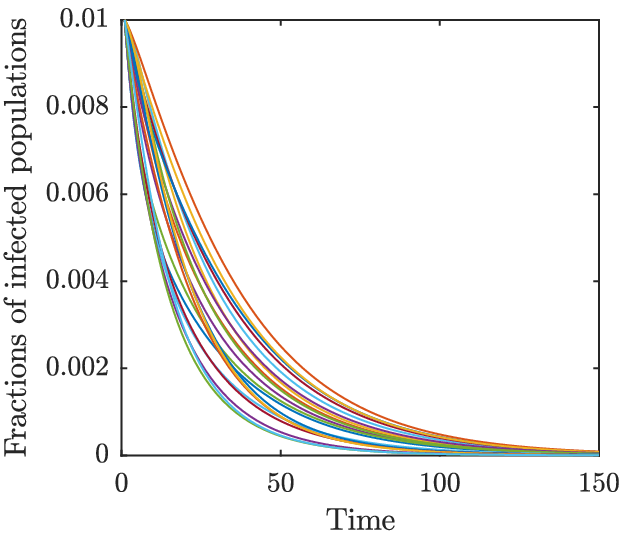}}\label{fig:Stochastic_b}
\subfigure[Stable endemic equilibrium: Stochastic model]{\includegraphics[width=0.23\textwidth]{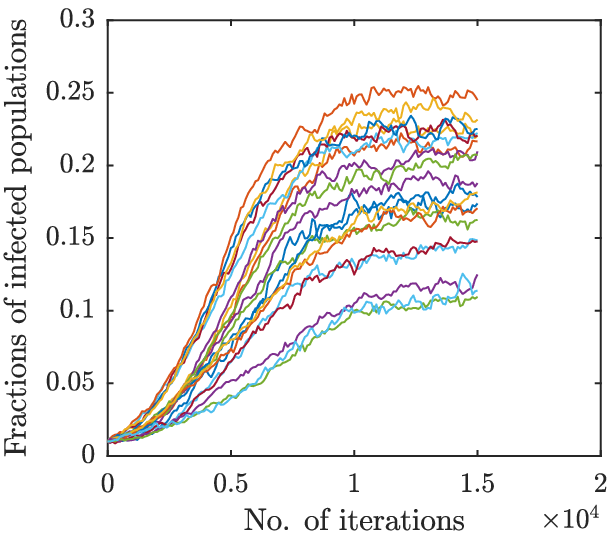}}\label{fig:Stochastic_c}
\subfigure[Stable endemic equilibrium: Deterministic model]{\includegraphics[width=0.23\textwidth]{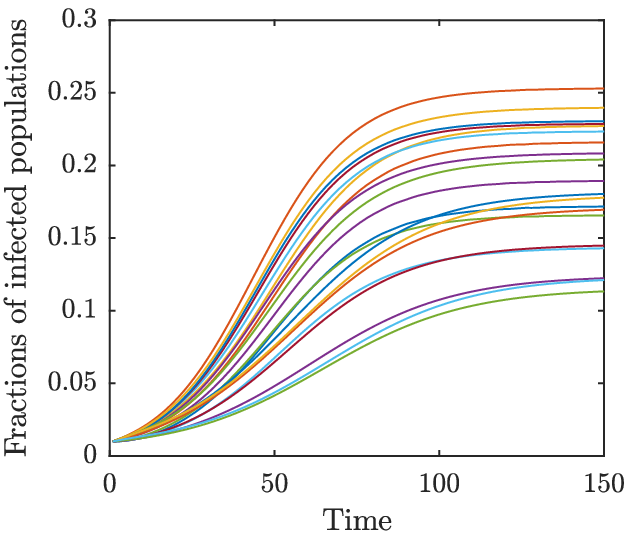}}\label{fig:Stochastic_d}
\caption{Stochastic simulation of epidemic spread under mobility. Line graph, $n=20$, $\nu(i) = 0.2$, $q_{i j}=\frac{\nu(i)}{D_{out}}$, $p_i(0)=0.01$. Each iteration in stochastic model corresponds to time-step $0.01$ sec.}
\label{fig:Stochastic}
\end{figure}
\endgroup
Once we have established the correctness of deterministic model predictions with the stochastic simulations, we study the simulations of deterministic model only. We study the effect of mobility over 4 different mobility graph structure - line graph, ring graph, star graph and a complete graph. First we keep the equilibrium distribution of population same for all the four graphs by using instantaneous transition rates from Metropolis-Hastings algorithm \cite{Hastings_MetroplisHastingsMC}. This shows the effect of different mobility graph structure on epidemic spread while the equilibrium population distribution remains the same. Fig.~\ref{fig:Deterministic_SameMobilityEqb} shows the fractions of infected population trajectories for $20$ nodes connected with $4$ different graph structures. The nodes have heterogeneous curing rates and these rates are the same across different graph structures. The values of equilibrium fractions are affected by the presence of mobility and are different for different graph structures. As seen in Fig.~\ref{fig:Deterministic_SameMobilityEqb}, star graph has the widest distribution of equilibrium infected fraction values whereas complete graph has the narrowest of the four.

\begingroup
\centering
\begin{figure}[ht!]
\centering
\subfigure[Line graph]{\includegraphics[width=0.23\textwidth]{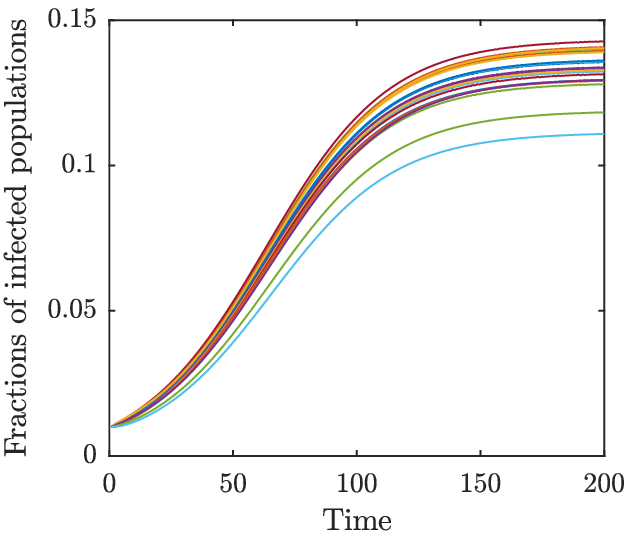}}\label{fig:SameMobilityEqbLine}
\subfigure[Ring graph]{\includegraphics[width=0.23\textwidth]{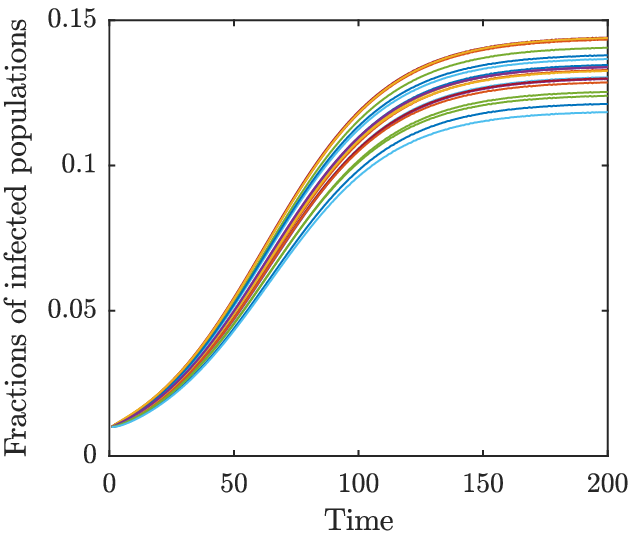}}\label{fig:SameMobilityEqbRing}
\subfigure[Star graph]{\includegraphics[width=0.23\textwidth]{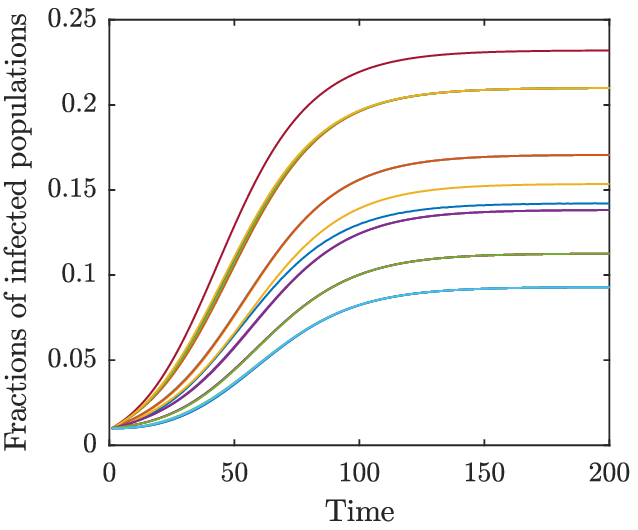}}\label{fig:SameMobilityEqbStar}
\subfigure[Complete graph]{\includegraphics[width=0.23\textwidth]{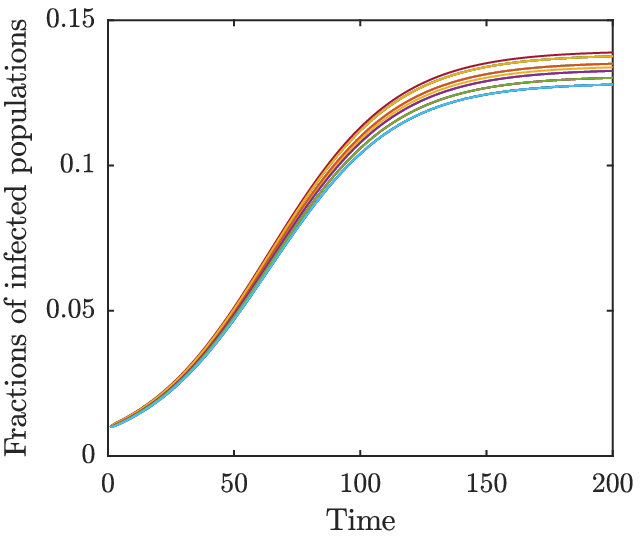}}\label{fig:SameMobilityEqbComplete}
\caption{Simulation of deterministic model of epidemic spread under mobility, with same equilibrium distribution of population over 4 different graph structure with stable endemic equilibrium. $n=20$, $p_i(0)=0.01$.}
\label{fig:Deterministic_SameMobilityEqb}
\end{figure}
\endgroup

Next, we verify the statement (iv) of Corollary \ref{cor:dis-free}, where one can have some curing rates $\delta_i$ less than the infection rates $\beta_i$ but still have stable disease-free equilibrium. We take a complete graph of $n=20$ nodes with given mobility transition rates which give us $\bs w$, $L^*$ and $\lambda_2$. We take a given set of values of $\beta_i$. Next, we compute $\subscr{m}{lower} = -\frac{\lambda_2}{4n+1}$ and take $0.8$ times of this value as $m$ in order to compute $\delta_i$'s that satisfy statement (iv) of Corollary \ref{cor:dis-free}. For our case the values are: $\beta_i=0.3$, $\lambda_2=0.2105$, $\subscr{m}{lower}=-0.0026$, $m = 0.8~ \subscr{m}{lower}=-0.0021$, $\delta_1=\delta_n=\beta_i+m$ and the rest $\delta_i$ computed to satisfy the condition which gives $\delta_1 = \delta_n = 0.2979$ and $\delta_i = 0.3198$ for $i \in \{2,\dots,n-1\}$. Fig.~\ref{fig:Lambda2 sufficient cond Complete graph} shows the trajectories of infected fraction populations. As can be seen the trajectories converge to the disease-free equilibrium.

\begin{figure}[ht!]
    \centering
    \includegraphics[width=0.9\linewidth]{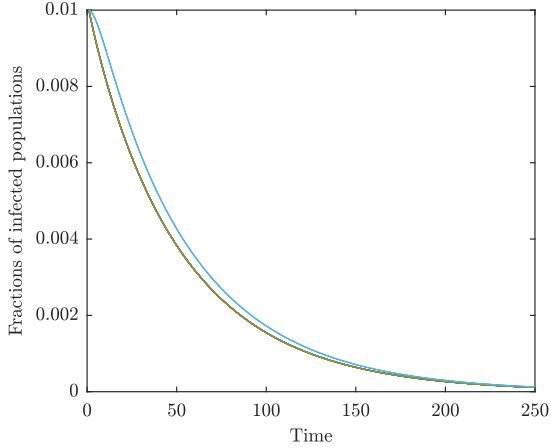}
    \caption{Stable disease-free equilibrium with curing rates computed as per the $\lambda_2$ sufficient condition (statement (iv), Corollary \ref{cor:dis-free}) for stability of disease-free equilibrium. Graph: Complete, $n=20$, $p_i(0)=0.01$.}
    \label{fig:Lambda2 sufficient cond Complete graph}
\end{figure}

\section{Conclusions} \label{Sec: conclusions}

We derived a continuous-time model for epidemic propagation under Markovian mobility across a network of sub-populations. The epidemic spread within each node has been modeled as SIS population model. The derived model has been analysed to establish the existence and stability of disease-free equilibrium and an endemic equilibrium under different conditions.
Some necessary and some sufficient conditions for stability of disease-free equilibrium have been established. We also provided numerical studies to support our results and elucidated the effect of mobility on epidemic propoagation.


\appendix
\subsection{Proof of Theorem 1 (iii): Existence of an endemic equilibrium} \label{Appendix: existence of non-trivial eqb}
We show below that in the case of $\mu (B-D-L^*) > 0$ , there exists an endemic equilibrium $\bs p^*$, i.e.,
\begin{equation} 
    \dot{\bs p} |_{\bs p =\bs  p^*} = (B-D-L^*- P^* B) \bs p^* = 0 .
\end{equation}

We use Brouwer's fixed point theorem, similar to the derivation in \cite{fall2007epidemiological}. Rearranging the terms and writing the above as an equation in $\bs p$ to be satisfied at non-trivial equilibrium $\bs p^*$ leads to:

\begin{equation} \label{eqAppendixeqb}
    (L^*+D)((L^*+D)^{-1} B - I)\bs p = P B \bs p .
\end{equation}
Define $A := (L^*+D)^{-1} B$. Since $A^{-1} = B^{-1} (L^*+D)$ is a non-singular M-matrix, its inverse $A$ is non-negative \cite{berman1994nonnegative}. Rearranging  \eqref{eqAppendixeqb} leads to 
\begin{equation}
   \bs p = H (\bs p) = (I + AP)^{-1}A \bs p .
\end{equation}

Now we show that $H(\bs p)$ as defined above is a monotonic function in the sense that $\bs p_{2} \geq \bs p_{1}$ implies $H(\bs p_{2}) \geq H(\bs p_{1})$. Define $\tilde{\bs p} := \bs p_2 - \bs p_1$ and $\tilde{P} := \operatorname{diag}(\tilde{\bs p})$. Then,
\begin{equation} \label{eqH}
    \begin{split}
          & \!\!\!\!\!   H(\bs p_{2}) - H(\bs p_{1}) \\
        &\!\!\!\!\! =  (A^{-1}+P_{2})^{-1}\bs p_{2} -  (A^{-1}+P_{1})^{-1}\bs p_{1} \\
        &\!\!\!\!\! = (A^{-1}+P_{2})^{-1}(\bs p_{2} - (A^{-1}+P_{2})(A^{-1}+P_{1})^{-1}\bs p_{1})\\
        & \!\!\!\!\!= (A^{-1}+P_{2})^{-1} (\tilde{\bs p} - \tilde{P}(A^{-1}+P_{1})^{-1}\bs p_{1})\\
        & \!\!\!\!\!= (A^{-1}+P_{2})^{-1} (I - \operatorname{diag}((A^{-1}+P_{1})^{-1}\bs p_{1}))\tilde{\bs p} .
    \end{split}
\end{equation}
Since $(A^{-1}+P_{2}) = B^{-1}(L^*+D) + P_{2}$ is an M-matrix its inverse and hence the first term above is non-negative. The second term is shown to be non-negative as below:
\begin{equation} \label{eqIAP}
    \begin{split}
    &(I - \operatorname{diag}((A^{-1}+P_{1})^{-1}\bs p_{1})) \\
        & = (I - \operatorname{diag}((I + A P_{1})^{-1} A P_{1} \bs 1_{n})) \\
        & = \operatorname{diag}((I - (I + A P_{1})^{-1} A P_{1}) \bs 1_{n}) \\
         & = \operatorname{diag}((I + A P_{1})^{-1}\bs 1_{n}) \\
         & = \operatorname{diag}((A^{-1} + P_{1})^{-1} A^{-1}\bs 1_{n}) \\
         & \geq 0 ,
    \end{split}
\end{equation}
where we have used the identity: 
\begin{equation}
    (I + X)^{-1} = I - (I+X)^{-1}X ,
\end{equation}
in the second line. The last inequality in \eqref{eqIAP} holds as $A^{-1} \bs 1_{n} = B^{-1} (L^* + D) \bs 1_{n} = B^{-1} D \bs 1_{n} \geq \bs 0_{n}$ and $(A^{-1} + P_{1})^{-1} \geq 0$ the inverse of an M-matrix. The last term in the last line of \eqref{eqH} is $\tilde{\bs p} \geq \bs 0_n$. This implies that $H(\bs p)$ is a monotonic function. Also, result in \eqref{eqIAP} implies that $H(\bs p) \leq \bs 1_n$ for all $\bs p \in [0,1]^n$. Therefore $H(\bs 1_n) \leq \bs 1_n$.\\

Convergent splitting property of irreducible M-matrices \cite{berman1994nonnegative} implies $\mu (B-D-L^*) > 0$ if and only if $R_0= \rho(A) = \rho((L^* +D)^{-1} B) > 1$. Here $\rho(A)$ is spectral radius of $A$. Since $A$ is an irreducible non-negative matrix, Perron-Frobenius theorem implies $\rho(A)$ is a simple eigenvalue with right eigenvector $\bs u$ satisfying $A \bs u = \rho (A) \bs u = R_0 \bs u$ , with $\bs u \gg \bs 0_n$. Define $U :=\operatorname{diag}(\bs u)$ and $\gamma := \frac{R_0 -1}{R_0}$. Now, we find a value of $\epsilon > 0$ such that $H(\epsilon \bs u)\geq \epsilon \bs{u}$ as below:
\begin{equation}
    \begin{split}
        H(\epsilon \bs u) - \epsilon \bs u & = (I+\epsilon AU)^{-1} A\epsilon \bs u - \epsilon \bs u \\
        & = (I - (I+\epsilon AU)^{-1} \epsilon AU) \epsilon R_0 \bs u -\epsilon \bs u \\
        & = \epsilon R_0 (\frac{(R_0-1)}{R_0} \bs u - (I+\epsilon AU)^{-1} \epsilon AU \bs u) \\
        & = \epsilon R_0 (\gamma \bs u - (I+\epsilon AU)^{-1} \epsilon AU \bs u) .
    \end{split}
\end{equation}

Now, the expression in the brackets in the last line is a continous function of $\epsilon$ and is equal to $\gamma \bs u \gg \bs 0_n $ at $\epsilon =0$. Therefore, there exists an $\epsilon > 0 $ such that $H(\epsilon \bs u) - \epsilon \bs u \geq \bs 0_n$ or equivalently, $H(\epsilon \bs u)\geq \epsilon \bs{u}$. Taking the closed compact set $K = [ \epsilon \bs u, \bs 1_n]$, $H: K \rightarrow K$ is a continuous function. Therefore, by Brouwer's fixed point theorem, there exists a fixed point in K. This proves the existence of a non-trivial equilibrium $\bs p^* \gg \bs 0_n$ when $\mu (B-D-L^*) > 0$ or equivalently $R_0 >1$. The uniqueness is further shown in the following proposition.

\begin{proposition}
If the mapping $H$ has a strictly positive fixed point, then it is unique.
\end{proposition}

\begin{proof}
The proof is similar to the proof of \cite[Proposition A.3]{khanafer_Basar2016stabilityEpidemicDirectedGraph} and is given below:\\
Assume there are two strictly positive fixed points: $\bs p^*$ and $\bs q^*$. Define 
\[\eta:=\max \frac{p^*_i}{q^*_i},\quad  k:= \arg\max
\frac{p^*_i}{q^*_i}\]
Therefore, $\bs p^* \leq \eta \bs q^*$. Lets assume $\eta >1$. First we will show that $H(\eta \bs q^*) < \eta H(\bs q^*)$ as follows:

\begin{equation}
\begin{split}
 & H(\eta \bs q^*)-\eta H(\bs q^*) \\
 & = (I + A\eta Q^*)^{-1}A \eta \bs q^* -\eta (I + A Q^*)^{-1}A \bs q^* \\
 &= ((I + A\eta Q^*)^{-1}A - (I + A Q^*)^{-1}A)\eta \bs q^* \\
  &= ((A^{-1} + \eta Q^*)^{-1} - (A^{-1} + Q^*)^{-1})\eta \bs q^* \\
  &= (A^{-1} + \eta Q^*)^{-1}(I - (A^{-1} + \eta Q^*)(A^{-1} + Q^*)^{-1})\eta \bs q^* \\
  &= (A^{-1} + \eta Q^*)^{-1}(-(\eta-1) (A^{-1} + Q^*)^{-1})\eta \bs q^*\\
  &< 0
\end{split}
\end{equation}
 where the last inequality uses result that inverse of a non-singular M-matrix is non-negative and non-singular, that $\eta > 1$ and, that $\bs q^* \gg \bs 0$ by assumption. Consequently
\begin{equation}
    p^*_k = H_k(\bs p^*)\leq H_k(\eta \bs q^*)< \eta H_k(\bs q^*) = \eta q^*_k,
\end{equation}
Since, $\eta q^*_k = p^*_k$ by definition, if $\eta > 1$, we have from above $p^*_k<p^*_k$, a contradiction. Hence, $\eta \leq 1$ which implies $\bs p^* \leq \bs q^*$. By switching the roles of $\bs p^*$ and $\bs q^*$ and repeating the above argument we can show $\bs q^* \leq \bs p^*$. Thus $\bs p^* = \bs q^*$ and hence there is a unique strictly positive fixed point.
\end{proof}

\end{document}